%% file: bare_jrnl_new_sample4.tex
\documentclass[lettersize,journal]{IEEEtran}
\usepackage{amsmath,amsfonts}
\usepackage{algorithm}
\usepackage{algorithmic}
\usepackage{array}
\usepackage[caption=false,font=footnotesize]{subfig}
\usepackage{textcomp}
\usepackage{stfloats}
\usepackage{url}
\usepackage{verbatim}
\usepackage{graphicx}
\usepackage{cite}

\usepackage{amssymb,amsthm}
\usepackage{adjustbox}
\usepackage{xcolor}
\newcommand*\D{\mathop{}\!\mathrm{d}}

\theoremstyle{plain}
\newtheorem{theorem}{Theorem}[section]

\newtheorem{definition}[theorem]{Definition}

\newtheorem{proposition}[theorem]{Proposition}
\newtheorem{corollary}[theorem]{Corollary}

\begin{document}

\title{R\'enyi Pufferfish Privacy with Gaussian-based Priors: From Single Gaussian to Mixture Model}

\author{
Wenjin Yang,
Ni Ding,~\IEEEmembership{Member,~IEEE},
Zijian Zhang$^*$,~\IEEEmembership{Senior Member,~IEEE},
Zhen Li,
Jing Sun,
Jincheng An,
Yong Liu$^*$,
Liehuang Zhu,~\IEEEmembership{Senior Member,~IEEE}

\thanks{Wenjin Yang, Zijian Zhang, Zhen Li and Liehuang Zhu are with the School of Cyberspace Science and Technology, Beijing Institute of Technology, 100081 Beijing, China. (email: \{wenjinyang, zhangzijian, zhen.li, liehuangz\}@bit.edu.cn)}
\thanks{Ni Ding, Jing Sun are with the School of Computer Science, University of Auckland, New Zealand. (email: jing.sun@auckland.ac.nz, dingni529@gmail.com)}
\thanks{Jincheng An is with QAX Security Center, the Qi-AnXin Technology Group Inc., Beijing, China (email: anjincheng@qianxin.com)}
\thanks{Yong Liu is with the Qi An Xin Technology Group Inc. and Zhongguancun Laboratory, Beijing 100000, China (email: liuyong03@qianxin.com)}
\thanks{Corresponding authors: Zijian Zhang, Yong Liu.}
\thanks{Manuscript received April 19, 2021; revised August 16, 2021.}}

\markboth{Journal of \LaTeX\ Class Files,~Vol.~14, No.~8, August~2021}%
{Shell \MakeLowercase{\textit{et al.}}: A Sample Article Using IEEEtran.cls for IEEE Journals}


\maketitle

\input{sec/0_abs.tex}

\begin{IEEEkeywords}
R\'enyi divergence, Pufferfish privacy, Gaussian distribution, Gaussian mixture model
\end{IEEEkeywords}

\input{sec/1_intro.tex}

\input{sec/2_prelim.tex}

\input{sec/3_gaussian_mechanism.tex}

\input{sec/4_gmm.tex}

\input{sec/5_experiment.tex}
\input{sec/6_conclusion.tex}

\bibliographystyle{IEEEtran}
\bibliography{ref}


 




\vfill

\end{document}

%% file: sec/0_abs.tex
\begin{abstract}
R\'{e}nyi Pufferfish Privacy (RPP) provides a R\'{e}nyi divergence-based privacy framework for correlated data, but existing $\infty$-Wasserstein mechanisms are often conservative and sacrifice data utility. 
We study Gaussian mechanisms for RPP under Gaussian and Gaussian-mixture priors. 
For single Gaussian priors, we derive the exact R\'{e}nyi divergence after Gaussian perturbation, obtain a relaxed closed-form sufficient condition for $(\alpha,\epsilon)$-RPP, and characterize the monotonicity of the calibrated noise with respect to the privacy budget $\epsilon$ and the R\'{e}nyi order $\alpha$. 
To handle more general non-Gaussian and multimodal priors, we approximate secret-conditioned outputs with Gaussian mixture models and introduce an optimal-transport-based sufficient condition for RPP. 
Experiments on three UCI datasets with statistical (\textsc{RAW}, \textsc{MEAN}) and model-output (\textsc{BNN}, \textsc{GP}) queries show that our prior-aware mechanisms consistently require less noise than a recent RPP additive-noise baseline, achieving an average noise reduction of 48.9\%. 
These results show that our mechanisms can substantially improve the privacy-utility trade-off under RPP.
\end{abstract}

%% file: sec/1_intro.tex
\section{Introduction}

\IEEEPARstart{T}{rustworthy} and privacy-preserving data circulation is a critical challenge these days. 
When participating in data analysis, we often need to release admissible information from the data while keeping the privacy. 
For example, when training medical LLMs, we want the LLMs to capture the general medical knowledge (e.g., disease diagnosis, treatment recommendations, etc.) while preserving the privacy of the patients (e.g., race, gender, age, etc.). 
Existing cryptographic techniques, including zero-knowledge proofs~\cite{ssc19,xing25}, homomorphic encryption~\cite{gentry09, hu24}, secure multi-party computation~\cite{yao82,zzw25}, and trusted execution environments~\cite{vs16, burke24}, provide strong protection under specific threat models, but they are often computationally expensive and rely on strict assumptions.
This motivates the development of formal privacy notions that quantify privacy risk and enable mechanism design with provable privacy guarantees.
To address this, data privacy metrics have been proposed, and Differential Privacy (DP) \cite{dwork06,cfka06} has become the de facto standard. 
DP guarantees that the output distribution of a randomized mechanism changes only slightly when the data of any single individual is modified or removed. 
Nevertheless, such a guarantee may be insufficient when the sensitive attributes and released outputs are statistically correlated.

Existing DP mechanisms often assume the query outputs attributes are independent of each other, which is not always true in real-world scenarios. 
In fact, the sensitive attributes and query outputs are often statistically dependent, which is called correlated data. 
In such a setting, privacy cannot always be characterized by neighboring-database notions \footnote{DP relies on the notion of neighboring-database $D$ and $D'$, which is typically chosen to capture the contribution to the mechanism's input by a single individual.} alone (the underlying setting of DP), because the adversary may exploit prior knowledge about the data distribution and perform powerful inference attacks~\cite{zlhg25, lxhg23}. 
Pufferfish Privacy (PP), a flexible privacy framework formulated by \cite{da12, da14}, addresses this issue by modeling the conditional distribution of the query outputs given the sensitive attributes and adversary's prior knowledge. 
To protect the sensitive data, pufferfish privacy enforces that the statistical distinguishability between two output probability distributions conditioned on a pair of secrets is upper bounded by a given privacy budget $\epsilon$.
Its R\'{e}nyi-divergence-based variant, R\'{e}nyi Pufferfish Privacy (RPP) \cite{cam24}, further replaces the measure with R\'{e}nyi divergence \cite{tr14}, enabling a fine-grained privacy characterization.

Despite the advantages of PP and RPP, mechanism design under this guarantee remains challenging. 
For PP, the first $\infty$-Wasserstein-based mechanism is proposed by \cite{swc17}. 
As the $\infty$-Wasserstein distance is not computable due to the non-convexity of the underlying minimization problem~\cite{cdj08,dl19}, \cite{ding22} proposes a $1$-Wasserstein (Kantorovich) mechanism, where the optimal transport plan (the minimizer) can be calculated directly by system parameters. 
\cite{ydz26} proposes a relaxed practical mechanism based on the $W_1$ mechanism and theoretically shows the noise reduction in a low privacy budget regime. 
For RPP, \cite{cam24} introduces the General Wasserstein Mechanism by relaxing the $\infty$-Wasserstein distance in \cite{swc17}, along with a $\delta$-approximation to improve data utility. 
\cite{zv25} defines Sliced R\'{e}nyi Pufferfish Privacy and proposes a sliced Wasserstein mechanism.
However, the first mechanism for RPP is often conservative and injects too much noise, which limits the utility of the mechanism, and the existing approaches improve data utility by relaxing the privacy guarantee, which is not always desirable.

In our work, we start from the R\'{e}nyi Pufferfish Privacy with a Gaussian prior. 
Gaussian prior is central to several active research directions, including machine learning, image/video reconstruction, generative modeling, etc. 
For example, in Bayesian neural networks, Gaussian priors are routinely imposed on network weights to regularize the model and quantify posterior uncertainty \cite{oa21}. 
In image reconstruction, Gaussian-process priors have been used to encode correlations in the image space or frequency space and support Bayesian reconstruction \cite{qwmg24}. 
In generative modeling, structured Gaussian priors have also been used in GP-prior VAEs for modeling correlations in high-dimensional and temporally structured datasets \cite{bsol25}. 
This makes the Gaussian-prior setting a natural and practically meaningful starting point for mechanism design under R\'enyi Pufferfish Privacy.
When a single Gaussian prior is too restrictive, the Gaussian mixture model provides a standard and flexible extension to capture general distributions. \cite{lfk18, pathak12}

For single Gaussian prior, we first derive the R\'{e}nyi divergence between two Gaussian mechanism outputs. 
To satisfy $(\epsilon, \alpha)$-RPP and preserve the data utility, we must determine the minimal variance of the Gaussian noise that satisfies the R\'{e}nyi divergence is less than privacy budget $\epsilon$. 
As the divergence is complex and cannot be solved directly, we analyze the monotonicity of the privacy loss with respect to the noise variance. 
The non-increasing property allows us to use binary search to find the minimal variance of the Gaussian noise that satisfies the R\'{e}nyi divergence is less than privacy budget $\epsilon$. 
To improve theoretically transparency and practically efficiency, we derive a closed-form sufficient condition by separatelly bound the two terms in the R\'{e}nyi divergence.
Then, we characterize the monotonicity of the calibrated noise with respect to the privacy budget $\epsilon$ and the R\'{e}nyi order $\alpha$, 
and show that the noise is non-increasing with respect to $\epsilon$ and the monotonicity with respect to $\alpha$ depends on the relative dominance between mean and variance of the Gaussian prior.
To handle more general non-Gaussian and multimodal priors, we model secret-conditioned outputs with Gaussian mixture models and introduce an optimal-transport-based sufficient condition for $(\epsilon,\alpha)$-RPP. 

The main contributions of this work are listed below.
\begin{enumerate}
  \item For a single Gaussian prior, we derive the exact R\'{e}nyi divergence between the outputs of Gaussian mechanisms, along with a theoretically-convergent numerical method to determine noise variance.
  \item By relaxing the transcendent condition of RPP, we propose a closed-form sufficient condition for $(\epsilon,\alpha)$-\textit{R\'{e}nyi Pufferfish Privacy}, which can be computed directly from the parameters of the Gaussian prior and the noise. We both theoretically and numerically show that the calibrated noise is decreasing with privacy budget, and the monotonicity with R\'{e}nyi order is determined by the relative dominance between mean and variance of the Gaussian prior.
  \item Extend to non-Gaussian priors by modeling outputs with Gaussian mixture models, we introduce an optimal-transport-based sufficient condition for $(\epsilon,\alpha)$-RPP.
  \item Experiments on three real-world datasets and both statistical and model-output queries show that the proposed mechanisms consistently require less noise than a recent additive-noise baseline, leading to improved utility under the same $(\epsilon,\alpha)$-RPP guarantee.
\end{enumerate}

\textbf{Related Works.}
Pufferfish Privacy was introduced by \cite{da12, da14} as a general privacy framework that explicitly models correlated data and adversarial prior knowledge.
Building on this framework, \cite{swc17} proposed the first Wasserstein-based mechanism for correlated data, while \cite{ding22} later introduced a Kantorovich ($W_1$)-based calibration method that makes standard PP substantially more tractable. 
More recently, \cite{ydz26} proposed a relaxed practical mechanism by relaxing the too strict $W_1$-based condition, and theoretically showed strict noise reduction, with especially pronounced reduction in the low-privacy-budget regime. 
In parallel, \cite{ding24} studied Gaussian-prior and GMM-based approximation under $(\epsilon, \delta)$-approximate PP, providing the first analysis of Gaussian/GMM priors in the PP setting. 
For R\'{e}nyi Pufferfish Privacy, \cite{cam24} generalized the Wasserstein mechanism to the R\'{e}nyi-divergence-based setting. 
Subsequent work \cite{zv25} relaxed the standard RPP to sliced RPP, along with a high-dimensional sliced mechanism. 
Other applications further develop this area, including Multi-user setting~\cite{dlyz25}, composition~\cite{bhgm26}, and quantum PP~\cite{ngw24, nsw25}. 
Consequently, our work complements this line of research by filling the gap of mechanism design for Gaussian and GMM priors under the RPP framework.

\textbf{Notation.}
We use $\mathcal{S}$ to denote the set of secrets, $\mathcal{Q}\subseteq\mathcal{S}\times\mathcal{S}$ to denote the set of discriminative secret pairs, and $\mathcal{P}$ to denote the class of adversarial prior beliefs. 
For a query output $X$ and a randomized mechanism $\mathcal{M}$, the privatized release is denoted by $Y=\mathcal{M}(X)=X+N_\theta$, where $N_\theta\sim\mathcal{N}(0,\theta^2)$ is additive Gaussian noise. 
We write $(\epsilon,\alpha)$-RPP for $(\alpha,\epsilon)$-\textit{R\'{e}nyi Pufferfish Privacy}, where $\alpha>1$ is the R\'{e}nyi order and $\epsilon>0$ is the privacy budget. 
In the single-Gaussian setting, the secret-conditioned priors are written as $X|s_i\sim\mathcal{N}(\mu_i,\sigma_i^2)$ and $X|s_j\sim\mathcal{N}(\mu_j,\sigma_j^2)$. 
In the GMM setting, we write $P(x|s)=\sum_{k=1}^{K_s}w_{s,k}\mathcal{N}(x;\mu_{s,k},\sigma_{s,k}^2)$, where $w_{s,k}$, $\mu_{s,k}$, and $\sigma_{s,k}^2$ denote the mixture weight, mean, and variance of the $k$-th component, respectively. 
When comparing two Gaussian mixtures, $\pi^*$ denotes the optimal transport coupling between their mixture components. 
Throughout the paper, a smaller calibrated noise $\theta$ (or equivalently $\theta^2$) indicates better utility under the same privacy guarantee.

\textbf{Organization.}
Section~\ref{sec:perlim} reviews preliminaries on Pufferfish Privacy, R\'{e}nyi Pufferfish Privacy, Gaussian mixture models, and optimal transport. 
Section~\ref{sec:gaussian_rpp} studies Gaussian mechanisms under single Gaussian priors, derives the exact R\'{e}nyi divergence after Gaussian perturbation, and presents both numerical and closed-form calibration results. 
Section~\ref{sec:gmm_rpp} extends the framework to Gaussian mixture priors through GMM fitting and an optimal-transport-based sufficient condition for RPP. 
Section~\ref{sec:experiments} evaluates the proposed mechanisms on real-world datasets using statistical and model-output queries. 
Finally, Section~\ref{sec:conclusion} concludes the paper.

%% file: sec/2_prelim.tex
\section{Preliminaries}\label{sec:perlim}
We review in this section the preliminaries and formulations needed to develop our work, including the Pufferfish Privacy and R\'{e}nyi Pufferfish Privacy framework, the Gaussian mixture model, and its optimal transport plan.

\textbf{Pufferfish Privacy Framework.}
Unlike Differential Privacy (DP), which considers neighboring databases, Pufferfish Privacy explicitly models the correlation under potential secrets and the adversary's prior knowledge about the data distribution. 
Let $\mathcal{S}$ denote the set of secrets, let $\mathcal{Q}\subseteq \mathcal{S}\times\mathcal{S}$ denote the set of discriminative secret pairs, and let $\mathcal{P}$ denote the class of prior beliefs.
A mechanism $\mathcal{M}$ satisfies $\epsilon$-\textit{Pufferfish Privacy} if for all $(s_i, s_j) \in \mathcal{Q}$ and all $\rho \in \mathcal{P}$,
\begin{align*}
    P(\mathcal{M}(x)|s_i, \rho) \leq e^\epsilon P(\mathcal{M}(x)|s_j, \rho).
\end{align*}
While this definition is robust, it can be more accurate by introducing a R\'{e}nyi divergence-based version, like the R\'{e}nyi Differential Privacy~\cite{mironov17}, which quantifies privacy guarantees by bounding certain moments of the exponential of the privacy loss. 

To achieve tighter privacy accounting, R\'{e}nyi divergence has been introduced into Pufferfish Privacy.
For two probability measures $P$ and $Q$ on a common measurable space and positive order $\alpha \in (0,1) \cup (1,\infty)$, the R\'{e}nyi divergence is defined as
\begin{align*}
    D_\alpha(P\|Q) = \frac{1}{\alpha-1} \ln \int p(x)^\alpha q(x)^{1-\alpha}\D x.
\end{align*}
R\'{e}nyi Pufferfish Privacy (RPP) framework generalizes the original framework by replacing the max-divergence with R\'{e}nyi divergence, allowing for a more accurate privacy analysis.
\begin{definition}[R\'{e}nyi Pufferfish Privacy, RPP]
    For $\alpha > 1$ and $\epsilon \geq 0$, a privacy mechanism $\mathcal{M}$ attains $(\alpha,\epsilon)$-\textit{R\'{e}nyi Pufferfish Privacy} if for all secret pairs $(s_i,s_j) \in \mathcal{Q}$ and all prior beliefs $\rho$, 
    \begin{align}
        D_\alpha(P(\mathcal{M}(x)|s_i, \rho) \| P(\mathcal{M}(x)|s_j, \rho)) \leq \epsilon,
    \end{align}
    where $D_\alpha$ is the $\alpha$-order R\'{e}nyi divergence.
\end{definition}
As $\alpha \to \infty$, RPP recovers the standard $\epsilon$-\textit{Pufferfish Privacy}. 
In this work, we focus on the Gaussian mechanism $\mathcal{M}(x) = f(x) + N_\theta$, which gives a way to calibrate a zero-mean Gaussian perturbation $N_\theta \sim \mathcal{N}(0, \theta^2)$ to the output of query $f(\cdot)$.

\textbf{Gaussian Mixture Model.}
The Gaussian mixture model (GMM) is a probabilistic model that can better characterize data distribution. For an arbitrary distributed $X|s_i$, we assume the adversary can train a Gaussian mixture model $X|s_i \sim \mathcal{GM}(D_i)$ from the original database $D_i$. 
Then, under the adversary's prior belief $\rho$, the probability distribution of $X$ given secret $s_i$ is 
\begin{align*}
    P(x|s_i,\rho) = \sum_{k=1}^{K_i} w_{ik} \mathcal{N}(x;\mu_{ik}, \sigma^2_{ik}),
\end{align*}
where $\sum_{k=1}^{K_i} w_{ik}=1$ and $\mathcal{N}(x;\mu_{ik}, \sigma^2_{ik})$ is the Gaussian component with mean $\mu_{ik}$ and variance $\sigma^2_{ik}$.

For another GMM denoted by $P(x'|s_j,\rho)$ with mean $\mu_{jl}$ and variance $\sigma^2_{jl}$, we introduce the optimal transport plan $\pi^*$ between these two GMMs, which determines the optimal way to transform one GMM into another. 
\begin{align*}
    \pi^{*}(x, x') = \sum_{k=1}^{K}\sum_{l=1}^{L}\pi_{k,l}^{*}\mathcal{N}(x; \mu_{ik}, \sigma_{ik}^{2}) \cdot \delta \left( x' - T_{k \to l}(x) \right),
\end{align*}
where $\pi_{k,l}^{*}$ represents the optimal mass transport from the $k$-th component
of $P(x|s_i,\rho)$ to the $l$-th component of $P(x'|s_j,\rho)$. The term $T_{k \to l}(x)$ denotes the optimal Monge map between the source component
$\mathcal{N}(x; \mu_{ik}, \sigma_{ik}^{2})$ and the target component
$\mathcal{N}(x'; \mu_{jl}, \sigma_{jl}^{2})$, which admits the closed-form expression:
\begin{align*}
    T_{k \to l}(x) = \mu_{jl}+ \frac{\sigma_{jl}}{\sigma_{ik}}(x - \mu_{ik}).
\end{align*}

%% file: sec/3_gaussian_mechanism.tex
\section{RPP Gaussian Priors}\label{sec:gaussian_rpp}
In this section, we propose the results for the Gaussian mechanism with Gaussian prior to obtain R\'{e}nyi Pufferfish Privacy guarantees and the numerical method to determine the variance of the additive noise. 
To improve computational efficiency and analytical interpretability, we also present a closed-form solution for calibrated noise variance. 

We first introduce R\'{e}nyi divergence between two Gaussian distributions \cite{tr14} and the Gaussian mechanism with Gaussian prior. 
\begin{definition}[Gaussian R\'{e}nyi divergence]
    For any simple order $\alpha$, the R\'{e}nyi divergence of a normal distribution $\mathcal{N}(\mu_0,\sigma_0^2)$ from another normal distribution $\mathcal{N}(\mu_1,\sigma_1^2)$ is 
    \begin{align}
        \label{eq:renyi_normal}
        D_\alpha\Big(
        \mathcal{N}(\mu_0,\sigma_0^2)
        \,\|\, \mathcal{N}(\mu_1,\sigma_1^2)
        \Big)
        &=
        \frac{\alpha(\mu_0-\mu_1)^2}{2\sigma_{\alpha}^2}
        \notag\\
        &\quad
        +\frac{1}{1-\alpha}\ln\frac{\sigma_\alpha}{\sigma_0^{1-\alpha} \sigma_1^\alpha}.
    \end{align}
    provided that $\sigma_\alpha^2 = (1-\alpha) \sigma_0^2 + \alpha \sigma_1^2$.
\end{definition}
Due to the addition property of the Gaussian distribution~\cite{mccool92}, we have 
\begin{proposition}[Gaussian mechanism with Gaussian prior]
    For additive Gaussian noise $N\sim\mathcal{N}(0,\theta^2)$ and Gaussian prior $X \sim \mathcal{N}(\mu,\sigma^2)$, the output $Y=N+X$ still follows a Gaussian distribution, i.e., $Y \sim \mathcal{N}(\mu,\sigma^2+\theta^2)$.
\end{proposition}

\subsection{Mechanism Design}
We now show that R\'{e}nyi divergence between two Gaussians can be directly used as the condition to attain $(\epsilon, \alpha)$-\textit{R\'{e}nyi Pufferfish Privacy} for the Gaussian mechanism with Gaussian prior.
\begin{corollary}\label{corollary:renyi_gaussian_mechanism}
    For all secret pairs $(s_{i},s_{j})\in\mathcal{Q}$ and Gaussian priors $x|s_{i} \sim \mathcal{N}(\mu_{i}, \sigma_{i}^{2})$ and $x'|s_{j} \sim \mathcal{N}(\mu_{j}, \sigma_{j}^{2})$, let $\epsilon>0$ and $\alpha>1$, adding Gaussian noise $\mathcal{N}(0,\theta^{2})$ with
    \begin{align}
        \label{eq:gaussian_mechanism}
        &\frac{\alpha(\mu_{i}-\mu_{j})^{2}}
        {2(\theta^{2}+(1-\alpha)\sigma_{i}^{2}+\alpha\sigma_{j}^{2})}
        \notag\\
        &\quad
        + \frac{1}{2(1-\alpha)}
        \ln \frac{\theta^{2}+(1-\alpha)\sigma_{i}^{2}+\alpha\sigma_{j}^{2}}
        {(\theta^{2}+\sigma_{i}^{2})^{1-\alpha}(\theta^{2}+\sigma_{j}^{2})^{\alpha}}
        \leq \epsilon
    \end{align}
    attains $(\epsilon, \alpha)$-\textit{R\'{e}nyi Pufferfish Privacy}.
\end{corollary}
\begin{proof}(Proof of Corollary~\ref{corollary:renyi_gaussian_mechanism})
    For any secret pair $(s_i,s_j)\in\mathcal{Q}$, the query outputs follows the Gaussian distribution, i.e., $X|s_i \sim \mathcal{N}(x; \mu_{i}, \sigma_{i}^{2})$ and $X|s_j \sim \mathcal{N}(x'; \mu_{j}, \sigma_{j}^{2})$. 
    The Gaussian mechanism calibrates a zero-mean Gaussian additive noise to query outputs as $Y=X+N_\theta$ where $N_\theta \sim \mathcal{N}(0, \theta^{2})$. Then the results of Gaussian mechanism with Gaussian prior follows a new Gaussian mechanism, i.e., $Y|s_i \sim \mathcal{N}(y; \mu_{i}, \sigma_{i}^{2}+\theta^2)$ and $Y|s_{j} \sim \mathcal{N}(y;\mu_{j}, \sigma_{j}^{2}+\theta^2)$. 
    Thus, the R\'{e}nyi divergence between $P(y|s_{i}, \rho)$ and $P(y|s_{j}, \rho)$ is 
    \begin{align*}
         & D_{\alpha}(P(y|s_{i},\rho) || P(y|s_{j},\rho) ) \\
         & \triangleq \frac{1}{\alpha-1}\ln \int P(y|s_{i},\rho)^{\alpha}P(y|s_{j},\rho)^{1-\alpha}\D y, \\
         & = \frac{1}{\alpha-1}\ln
         \frac{1}{(2\pi)^{1/2} (\theta^2+\sigma_i^2)^{\alpha/2}
         (\theta^2+\sigma_j^2)^{(1-\alpha)/2}}
         \notag\\
         &\quad
         \int \exp \left(
         \frac{\alpha-1}{\theta^2+\sigma_j^2} \frac{(y-\mu_{j})^{2}}{2}
         - \frac{\alpha}{\theta^2+\sigma_i^2} \frac{(y-\mu_i)^{2}}{2}
         \right) \D y,\\
         & = \frac{1}{\alpha-1} \ln \Big(
         \sqrt{\frac{(\sigma_{i}^{2}+\theta^{2})^{1-\alpha}
         ({\sigma_{j}}^{2}+\theta^{2})^{\alpha}}
         {\alpha{\sigma_{j}}^{2}+(1-\alpha)\sigma_{i}^{2}+\theta^{2}}} \\
         & \quad
         \times \exp\left(\frac{\alpha(\alpha-1)(\mu_i - \mu_j)^{2}}
         {2(\alpha{\sigma_j}^{2}+(1-\alpha)\sigma_i^{2}+\theta^{2})}\right)
         \Big), \\
         & = \frac{\alpha(\mu_{i}-\mu_{j})^{2}}
         {2(\theta^{2}+(1-\alpha)\sigma_{i}^{2}+\alpha\sigma_{j}^{2})}
         \notag\\
         &\quad
         + \frac{1}{2(1-\alpha)}\ln
         \frac{\theta^{2}+(1-\alpha)\sigma_{i}^{2}+\alpha\sigma_{j}^{2}}
         {(\theta^{2}+\sigma_{i}^{2})^{1-\alpha}(\theta^{2}+\sigma_{j}^{2})^{\alpha}}.
    \end{align*}
    To attain $(\alpha,\epsilon)$-\textit{R\'{e}nyi Pufferfish Privacy}, the divergence should be bounded by privacy budget $\epsilon$,
    \begin{align*}
        &\frac{\alpha(\mu_{i}-\mu_{j})^{2}}
        {2(\theta^{2}+(1-\alpha)\sigma_{i}^{2}+\alpha\sigma_{j}^{2})}
        \\
        &\quad
        + \frac{1}{2(1-\alpha)}\ln
        \frac{\theta^{2}+(1-\alpha)\sigma_{i}^{2}+\alpha\sigma_{j}^{2}}
        {(\theta^{2}+\sigma_{i}^{2})^{1-\alpha}(\theta^{2}+\sigma_{j}^{2})^{\alpha}}
        \leq \epsilon.
    \end{align*}
    Corollary~\ref{corollary:renyi_gaussian_mechanism} is proved.
\end{proof}
To satisfy the $(\epsilon, \alpha)$-RPP guarantee and preserve the utility of the mechanism, we must determine the minimal noise variance $\theta^2$ that satisfies the inequality \eqref{eq:gaussian_mechanism} for all secret pairs $(s_i, s_j) \in \mathcal{Q}$, rather than the feasible variance.
To achieve this, we first establish the monotonicity of the privacy loss function (LHS of Eq.~\eqref{eq:gaussian_mechanism}) with respect to $\theta^2$.
Let $g(\theta^2)$ denote the privacy loss function (the left-hand side of Eq. \eqref{eq:gaussian_mechanism}) and let $z = \theta^2$ represent the noise variance. 
We have
\begin{align}
    g(z)
    &\triangleq
    \frac{\alpha \Delta}{2(z + A(\alpha))}
    + \frac{1}{2(1-\alpha)} \ln
    \frac{z + A(\alpha)}
    {(z+\sigma_{i}^{2})^{1-\alpha}(z+\sigma_{j}^{2})^{\alpha}},
\end{align}
where $A(\alpha) = (1-\alpha) \sigma_i^2 + \alpha \sigma_j^2$ and $\Delta = (\mu_i - \mu_j)^2$.
For the R\'{e}nyi divergence to be well-defined, $z + A > 0$ should be satisfied, where $\theta^2 + (1-\alpha)\sigma_i^2 + \alpha\sigma_j^2 > 0$.
\begin{proposition}[Monotonicity of $g(\cdot)$]\label{prop:mono_privacy_loss}
    For $\alpha > 1$, the privacy loss function $g(\cdot)$ is monotonically non-increasing with respect to the Gaussian noise variance $\theta^2$.
\end{proposition}
\begin{proof}(Proof of Proposition~\ref{prop:mono_privacy_loss})
   Expanding $g(z)$ using logarithmic properties and differentiating with respect to $x$, we obtain
   \begin{align*}
        g'(z)
        =
        -\frac{\alpha\Delta}{2(z+A(\alpha))^2}
        + \frac{(\sigma_i^2 - \sigma_j^2)^2}
        {2(z + A(\alpha))(z + \sigma_i^2)(z + \sigma_j^2)}.
   \end{align*}
   Then, we have
   \begin{align*}
        g'(z)
        =
        -\frac{\alpha}{2(z+A(\alpha))} \left[
        \frac{\Delta}{z + A(\alpha)}
        + \frac{(\sigma_i^2 - \sigma_j^2)^2}
        {(z + \sigma_i^2)(z + \sigma_j^2)}
        \right].
   \end{align*}
   Given $\alpha > 1$ and the domain constraint $z + A(\alpha) > 0$, the coefficients and all terms within the brackets are non-negative. Consequently, $g'(z) \leq 0$ for all valid $z$, ensuring that $g(z)$ is monotonically non-increasing.
\end{proof}
The monotonicity guarantees that for any privacy budget $\epsilon > 0$ and R\'{e}nyi order $\alpha>1$, there exists a unique minimal noise variance $\theta^2$ that satisfies $(\epsilon,\alpha)$-\textit{R\'{e}nyi Pufferfish Privacy} requirement. 
Furthermore, it transforms the problem of finding the minimal $\theta^2$ to finding the root of $g(\theta^2) = \epsilon$. 
Since Eq.~\eqref{eq:gaussian_mechanism} does not admit a closed-form solution of $\theta^2$, we determine its root by numerical root-finding techniques (e.g., binary search) given that $g(\cdot)$ is strictly monotonically decreasing in $\theta^2$.

\subsection{Determining Gaussian Variance}
Before addressing the general case, we consider the symmetric prior scenario where the variances are equal, i.e., $\sigma_i^2 = \sigma_j^2 = \sigma^2$. In this setting, the logarithmic term in $g(\theta^2)$ is simplified as $(\theta^2 + \sigma^2)^{1-\alpha}(\theta^2 + \sigma^2)^\alpha$ equals $\theta^2 + \sigma^2$. 
The privacy requirement simplifies to
\begin{align}
   \frac{\alpha(\mu_i - \mu_j)^2}{2(\theta^2 + \sigma^2)}
   \leq \epsilon
   \quad \Rightarrow \quad
   \theta^2 \geq \frac{\alpha(\mu_i - \mu_j)^2}{2\epsilon} - \sigma^2.
   \label{eq:closed_form_sigma_example}
\end{align}
This closed-form solution reveals that the required noise power scales linearly with the squared mean difference $(\mu_i - \mu_j)^2$ and the order $\alpha$, and is inversely proportional to the privacy budget $\epsilon$, adjusted by the prior's inherent variance $\sigma^2$.

For the general case where $\sigma_{i}$ and $\sigma_{j}$ are not equal, if the R\'{e}nyi divergence naturely satisfies
\begin{align*}
    D_\alpha(P(x|s_i,\rho) \| P(x|s_j,\rho)) \leq \epsilon,
\end{align*}
the additive Gaussian noise is unnecessary. When this divergence is larger than privacy
budget $\epsilon$, due to the transcendental nature of privacy loss function $g(\cdot)$, there is no closed-form solution for $\theta^2$. 
However, since the R\'{e}nyi divergence $D_\alpha$ is monotonically decreasing with respect to the noise variance $\theta^2$ (Proposition~\ref{prop:mono_privacy_loss}), we can find the numerical root of $g(\theta^2) = \epsilon$ using root-finding methods. 
Specifically, we employ a binary search over a search interval $[0, \theta_{\text{max}}^2]$, where $\theta_{\text{max}}^2$ can be safely initialized using the symmetric case estimate. 
The algorithm is summarized in Algorithm~\ref{alg:theta_selection}.

\begin{algorithm}[!ht]
    \renewcommand{\algorithmicrequire}{\textbf{Input:}}
    \renewcommand{\algorithmicensure}{\textbf{Output:}}
    \caption{Determination of Noise Variance $\theta^2$}\label{alg:theta_selection}
    \begin{algorithmic}[1]
    \REQUIRE Secret parameters $(\mu_i, \sigma_i^2)$, $(\mu_j, \sigma_j^2)$, order $\alpha$, budget $\epsilon$, tolerance $\delta_{\text{tol}}$
    \ENSURE Minimal noise variance $\theta^2$
    \STATE Initialize $L \gets 0$ and $R \gets \frac{\alpha(\mu_i - \mu_j)^2}{2\epsilon}$
    \WHILE{$R - L > \delta_{\text{tol}}$}
    \STATE $M \gets (L + R) / 2$
    \STATE Compute $val \gets g(M)$ using Eq.~\eqref{eq:gaussian_mechanism}
    \STATE \textbf{if} $val > \epsilon$ \textbf{then} $L \gets M$ 
    \STATE \textbf{else} $R \gets M$
    \ENDWHILE
    \STATE \textbf{return} $R$
    \end{algorithmic}
\end{algorithm}

\subsection{Closed-form Gaussian Mechanism}
The main problems of determining Gaussian variance in Corollary~\ref{corollary:renyi_gaussian_mechanism} and Alg.~\ref{alg:theta_selection} lie in both computation and analysis, 
(a) Iterative root-finding algorithm leads to high computational complexity, which severely limits scalability in applications, e.g., training large machine learning models. 
(b) It's unclear to analyze the propositions of $\theta$ due to the transcendental nature of $g(\cdot)$.
Consequently, it is both theoretically and practically critical to derive a closed-form solution for $\theta^2$ to ensure $O(1)$ computational efficiency, analytical transparency, and absolute privacy guarantees.

The key difficulty in Corollary~\ref{corollary:renyi_gaussian_mechanism} is that the exact R\'{e}nyi privacy condition is transcendental in the noise variance $\theta^2$, since $\theta^2$ appears both in the rational term and inside the logarithm. 
Therefore, instead of solving the exact equality, we derive a closed-form sufficient condition by constructing an upper bound on the exact divergence.
Our relaxation keeps the mean-separation term exact, and relaxes only the logarithmic term induced by variance mismatch. 
This design preserves the main dependence on the prior means while converting the exact privacy condition into a quadratic inequality in the shifted noise variable. 
The results are shown as follows.

\begin{theorem}(Closed-form Mechanism)\label{theorem:closed_form_mechanism}
    Given privacy budget $\epsilon>0$ and R\'{e}nyi divergence order $\alpha>1$, for Gaussian prior $\mathcal{N}(\mu_i, \sigma_i^2)$ and $\mathcal{N}(\mu_j, \sigma_j^2)$ under all secret pairs $(s_i, s_j) \in \mathcal{Q}$, adding zero-mean Gaussian noise with variance
    \begin{align}\label{eq:closed_form_mechanism}
        \theta^2
        \geq
        \frac{\alpha}{4\epsilon}
        \left( Eq.A + \sqrt{\left( Eq.A \right)^2 + Eq.B} \right)
        - \sigma_i^2
    \end{align}
    attains $(\epsilon,\alpha)$-R\'{e}nyi Pufferfish Privacy, where $Eq.A = (\mu_i - \mu_j)^2 - 2\epsilon(\sigma_j^2 - \sigma_i^2)$ and $Eq.B = \frac{8\epsilon (\sigma_j^2 - \sigma_i^2)^2}{\alpha-1}$.
\end{theorem}
\begin{proof}(Proof of Theorem~\ref{theorem:closed_form_mechanism})
    Consider
    \begin{align*}
         & D_{\alpha}(P(y|s_{i},\rho) || P(y|s_{j},\rho) )\\
         & = \underbrace{\frac{\alpha(\mu_{i}-\mu_{j})^{2}}
         {2(\theta^{2}+(1-\alpha)\sigma_{i}^{2}+\alpha\sigma_{j}^{2})}}_{T_1}
         \notag\\
         &\quad
         + \underbrace{\frac{1}{2(1-\alpha)}\ln
         \frac{\theta^{2}+(1-\alpha)\sigma_{i}^{2}+\alpha\sigma_{j}^{2}}
         {(\theta^{2}+\sigma_{i}^{2})^{1-\alpha}(\theta^{2}+\sigma_{j}^{2})^{\alpha}}}_{T_2},
    \end{align*}
    Denote $\eta_i = \theta^2 + \sigma_i^2$ and $\eta_j = \theta^2 + \sigma_j^2$ where $\theta^2$ is the variance of additive Gaussian noise, we have
    \begin{align*}
        T_2 
        & = \frac{1}{2(\alpha-1)} \ln \frac{\eta_i^{1-\alpha} \eta_i^\alpha}{\eta_i+\alpha(\eta_j-\eta_i)},\\
        & = \frac{1}{2(\alpha-1)} \ln \frac{ (\frac{\eta_j}{\eta_i})^\alpha }{ \frac{\eta_i+\alpha(\eta_j-\eta_i)}{\eta_i} },\\
        & = \frac{1}{2(\alpha-1)} ( \alpha\ln\frac{\eta_j}{\eta_i} - \ln(1+\frac{\alpha(\eta_j-\eta_i)}{\eta_i}) ).
    \end{align*}
    Denote $\nu = \frac{\eta_j-\eta_i}{\eta_i}$, we have
    \begin{align*}
        T_2
        & = \frac{1}{2(\alpha-1)} ( \alpha\ln(1+\nu) - \ln(1+\alpha\nu) ),\\
        & \leq \frac{1}{2(\alpha-1)} (\alpha\nu - \frac{\alpha\nu}{1+\alpha\nu}),\\
        & = \frac{\alpha^2 (\eta_j - \eta_i)^2}{2(\alpha-1) \eta_i (\eta_i + \alpha(\eta_j - \eta_i))}.\\
    \end{align*}
    Thus, we have
    \begin{align*}
          & D_{\alpha}\big(P(y|s_{i},\rho)\,\|\,P(y|s_{j},\rho)\big)
         \leq \\
         & 
         \frac{1}{2(\eta_i + \alpha(\eta_j - \eta_i))}
         \left(
         \alpha(\mu_{i}-\mu_{j})^{2}
         + \frac{\alpha^2 (\eta_j - \eta_i)^2}{(\alpha-1)\eta_i}
         \right).
    \end{align*}
    To attain pufferfish privacy, we have
    \begin{align*}
        \frac{1}{2(\eta_i + \alpha(\eta_j - \eta_i))}
        \left(
        \alpha(\mu_{i}-\mu_{j})^{2}
        + \frac{\alpha^2 (\eta_j - \eta_i)^2}{(\alpha-1)\eta_i}
        \right)
        \leq \epsilon, \\
    \end{align*}
    Denote $\delta = \sigma_j^2-\sigma_i^2$, $\Delta = (\mu_i-\mu_j)^2$ and $x=\eta_i$, we have
    \begin{align*}
        & \frac{1}{2(x + \alpha \delta)}
        \left[
        \alpha \Delta + \frac{\alpha^2 \delta^2}{(\alpha-1)x}
        \right]
        \leq \epsilon, \\
        \Rightarrow & 2\epsilon x^2 + (\alpha \cdot 2\epsilon \delta - \alpha \Delta)x - \frac{\alpha^2 \delta^2}{\alpha-1} \geq 0.
    \end{align*}
    We have
    \begin{align*}
        \theta^2
        \geq
        & 
        \frac{\alpha}{4\epsilon}
        \Big(
        (\mu_i - \mu_j)^2 - 2\epsilon(\sigma_j^2 - \sigma_i^2)
        \notag\\
        & + \sqrt{\left( (\mu_i - \mu_j)^2 - 2\epsilon(\sigma_j^2 - \sigma_i^2) \right)^2
        + \frac{8\epsilon (\sigma_j^2 - \sigma_i^2)^2}{\alpha-1}}
        \Big)
        - \sigma_i^2
    \end{align*}
\end{proof}
Consider the special case where the two Gaussian priors only differ in mean, i.e., $\sigma_i = \sigma_j = \sigma$, adding Gaussian noise $\mathcal{N}(0,\theta^2)$ with
\begin{align*}
    \theta^2 \geq \frac{\alpha (\mu_i - \mu_j)^2}{2\epsilon} - \sigma^2,
\end{align*}
which is the same as the result in Eq.~\eqref{eq:closed_form_sigma_example}.

For the closed-form of noise variance, we analyze how it changes with the privacy budget $\epsilon$. A straightforward insight is that the noise will be smaller with the privacy budget $\epsilon$ increases, which represents the low privacy guarantee. We formally prove this as follows.
\begin{proposition}\label{prop:theta_decrease_epsilon}
    For R\'{e}nyi divergence order $\alpha>1$, the variance of Gaussian noise $\theta^2$ is decreasing with the privacy budget $\epsilon$.
\end{proposition}
\begin{proof}(Proof of Proposition~\ref{prop:theta_decrease_epsilon})
    Denote $\Delta=(\mu_i-\mu_j)^2$ and $\delta=\sigma_j^2-\sigma_i^2$. The
    non-constant part of Eq.~\eqref{eq:closed_form_mechanism} can be rationalized as
    \begin{align*}
        h(\epsilon)
        &=\frac{\alpha}{4\epsilon}
        \left(\Delta-2\epsilon\delta+
        \sqrt{(\Delta-2\epsilon\delta)^2+\frac{8\epsilon\delta^2}{\alpha-1}}\right)\\
        &=\frac{2\alpha\delta^2}
        {(\alpha-1)\left(\sqrt{(\Delta-2\epsilon\delta)^2+
        \frac{8\epsilon\delta^2}{\alpha-1}}-(\Delta-2\epsilon\delta)\right)} .
    \end{align*}
    Therefore, it suffices to show that
    \begin{align*}
        D(\epsilon)=
        \sqrt{(\Delta-2\epsilon\delta)^2+\frac{8\epsilon\delta^2}{\alpha-1}}
        -(\Delta-2\epsilon\delta)
    \end{align*}
    is increasing. Direct differentiation gives
    \begin{align*}
        D'(\epsilon)
        =
        2\delta\left(
        1+\frac{-(\Delta-2\epsilon\delta)+\frac{2\delta}{\alpha-1}}
        {\sqrt{(\Delta-2\epsilon\delta)^2+\frac{8\epsilon\delta^2}{\alpha-1}}}
        \right).
    \end{align*}

    If $\delta>0$, then $D'(\epsilon)>0$ follows from
    \begin{align*}
        \sqrt{(\Delta-2\epsilon\delta)^2+\frac{8\epsilon\delta^2}{\alpha-1}}
        >
        \Delta-2\epsilon\delta-\frac{2\delta}{\alpha-1},
    \end{align*}
    which is immediate when the right-hand side is negative; otherwise, after
    squaring, it reduces to
    \begin{align*}
        \frac{4\Delta\delta}{\alpha-1}
        >
        \frac{4\delta^2}{(\alpha-1)^2},
    \end{align*}
    which holds under the same nonnegative-right-hand-side condition. If
    $\delta<0$, then $D'(\epsilon)>0$ is equivalent to
    \begin{align*}
        \Delta-2\epsilon\delta-\frac{2\delta}{\alpha-1}
        >
        \sqrt{(\Delta-2\epsilon\delta)^2+\frac{8\epsilon\delta^2}{\alpha-1}},
    \end{align*}
    and squaring gives
    \begin{align*}
        -\frac{4\Delta\delta}{\alpha-1}
        +\frac{4\delta^2}{(\alpha-1)^2}>0,
    \end{align*}
    which is true since $\Delta\geq0$, $\delta<0$, and $\alpha>1$. Thus
    $D(\epsilon)$ is increasing for $\delta\neq0$, so $h(\epsilon)$ is decreasing.
    When $\delta=0$, Eq.~\eqref{eq:closed_form_mechanism} reduces to
    $\alpha\Delta/(2\epsilon)$, which is also decreasing in $\epsilon$.
\end{proof}

Consider the R\'{e}nyi divergence order, we also analyze the monotonicity of $\theta^2$ with respect to $\alpha$.
\begin{proposition}\label{prop:theta_monotonicity_alpha}
    For any $\epsilon>0$, the variance of Gaussian noise $\theta^2$ is increasing with the R\'{e}nyi divergence order $\alpha$ if 
    \begin{align}
        (\mu_i - \mu_j)^2 > 2\epsilon(\sigma_j^2 - \sigma_i^2) + \frac{2-\alpha}{\alpha-1} \sqrt{2\epsilon} |\sigma_j^2 - \sigma_i^2|,
    \end{align}
    else $\theta^2$ is decreasing with $\alpha$.
\end{proposition}
\begin{proof}(Proof of Proposition~\ref{prop:theta_monotonicity_alpha})
    Taking the partial derivative of $\theta^2$ with respect to the R\'{e}nyi divergence order $\alpha$ yields a fractional expression. For $\alpha > 1$, the leading coefficient and the denominator of the derivative are strictly positive. Through algebraic rearrangement and factoring out these positive terms, the sign of $\frac{\partial \theta^2}{\partial \alpha}$ is found to be solely determined by the following expression:
    \begin{align*}
        & \alpha \Big( (\mu_i - \mu_j)^2 - 2\epsilon(\sigma_j^2 - \sigma_i^2) \Big) + \\ & (\alpha-2)\sqrt{\left( (\mu_i - \mu_j)^2 - 2\epsilon(\sigma_j^2 - \sigma_i^2) \right)^2 + \frac{8\epsilon (\sigma_j^2 - \sigma_i^2)^2}{\alpha-1}}.
    \end{align*}
    To ensure that $\theta^2$ is strictly increasing with $\alpha$, we require $\frac{\partial \theta^2}{\partial \alpha} > 0$. Rearranging the terms, this condition becomes:
    \begin{align*}
        & (\mu_i - \mu_j)^2 - 2\epsilon(\sigma_j^2 - \sigma_i^2) >  \\ & \frac{2-\alpha}{\alpha} \sqrt{\left( (\mu_i - \mu_j)^2 - 2\epsilon(\sigma_j^2 - \sigma_i^2) \right)^2 + \frac{8\epsilon (\sigma_j^2 - \sigma_i^2)^2}{\alpha-1}}.
    \end{align*}
    To resolve the square root, we square both sides. We carefully account for the sign of the coefficient $\frac{2-\alpha}{\alpha}$, which differs depending on whether $1 < \alpha \le 2$ or $\alpha > 2$. In both cases, after substituting the squared terms and canceling out the common factors, the inequality consistently simplifies to:
    \begin{align*}
        \Big( (\mu_i - \mu_j)^2 - 2\epsilon(\sigma_j^2 - \sigma_i^2) \Big)^2 > \frac{(2-\alpha)^2}{(\alpha-1)^2} 2\epsilon (\sigma_j^2 - \sigma_i^2)^2.
    \end{align*}
    Taking the square root of both sides, and noting that the right side is strictly non-negative, we obtain:
    \begin{align*}
        (\mu_i - \mu_j)^2 - 2\epsilon(\sigma_j^2 - \sigma_i^2) > \frac{2-\alpha}{\alpha-1} \sqrt{2\epsilon} |\sigma_j^2 - \sigma_i^2|.
    \end{align*}
    This directly recovers the condition in the proposition. Conversely, if the inequality is reversed, the core expression is negative, yielding $\frac{\partial \theta^2}{\partial \alpha} < 0$, which means $\theta^2$ is decreasing with $\alpha$. This concludes the proof.
\end{proof}
Proposition~\ref{prop:theta_monotonicity_alpha} shows that the effect of the R\'{e}nyi order $\alpha$ on the required noise is not universally monotone. 
Instead, it depends on the relative dominance between mean and variance of the two secret-conditioned Gaussian priors. 
When mean dominates, increasing $\alpha$ makes the privacy requirement harder to satisfy and thus requires more noise. 
When variance dominates, the opposite trend may occur.

\newcommand{\gaussianthetafig}[3]{
    \subfloat[#2\label{#3}]{%
        \begin{minipage}[t]{0.48\columnwidth}
        \centering
        \begin{adjustbox}{width=\linewidth,max totalheight=0.17\textheight,keepaspectratio}
        \includegraphics[width=\linewidth]{figs/evaluation/#1.pdf}%
        \end{adjustbox}
        \end{minipage}%
    }
}

{\bf Visualization. }
Figure~\ref{fig:gaussian_query_sweeps} visualizes the calibrated noise scale for the Gaussian-prior setting and compares our closed-form calibration with the baseline~\cite{cam24}, which is an additive Gaussian mechanism based on $\infty$-Wasserstein distance. 
For both \textsc{Raw} and \textsc{Mean} queries, the required noise parameter $\theta$ decreases as the privacy budget $\epsilon$ increases, which is consistent with Proposition~\ref{prop:theta_decrease_epsilon}. 
Our mechanism also consistently requires smaller noise than the baseline, illustrating the utility gain of the closed-form calibration in Theorem~\ref{theorem:closed_form_mechanism}.
Figure~\ref{fig:theta_eps_gaussian_alpha} illustrates the effect of the R\'{e}nyi order $\alpha$ on the calibrated noise parameter $\theta$ for two Gaussian-prior examples. 
The case in Figure~\ref{fig:theta_eps_gaussian_alpha_dec} corresponds to the $\alpha$-decreasing regime in Proposition~\ref{prop:theta_monotonicity_alpha}, where larger $\alpha$ reduces the required noise, 
whereas the case in Figure~\ref{fig:theta_eps_gaussian_alpha_inc} corresponds to the $\alpha$-increasing regime, where larger $\alpha$ raises the required noise. 
Together, these examples confirm the two cases characterized by Proposition~\ref{prop:theta_monotonicity_alpha}.

\begin{figure}[!ht]
    \centering
    \gaussianthetafig{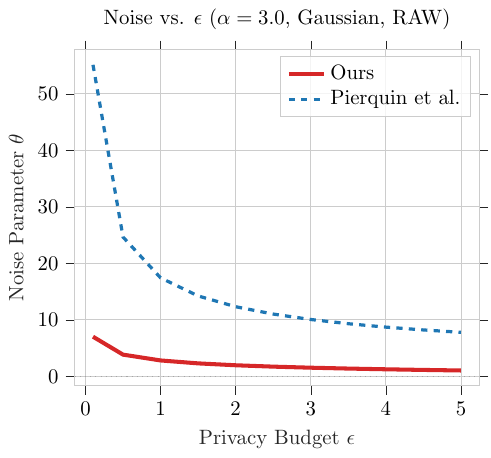}{\textsc{Raw}}{fig:theta_eps_gaussian_raw}
    \hfill
    \gaussianthetafig{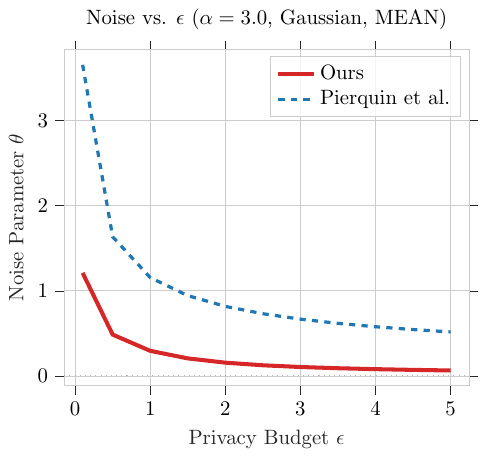}{\textsc{Mean}}{fig:theta_eps_gaussian_mean}
    \caption{Noise parameter $\theta$ versus privacy budget $\epsilon$ for Gaussian-prior queries with $\alpha=3.0$.}
    \label{fig:gaussian_query_sweeps}
\end{figure}

\newcommand{\gaussianalphafig}[3]{
    \subfloat[#2\label{#3}]{%
        \begin{minipage}[t]{0.95\columnwidth}
        \centering
        \begin{adjustbox}{width=\linewidth,max totalheight=0.18\textheight,keepaspectratio}
        \includegraphics[width=\linewidth]{figs/alpha/#1.pdf}%
        \end{adjustbox}
        \end{minipage}%
    }
}

\begin{figure}[!ht]
    \centering
    \gaussianalphafig{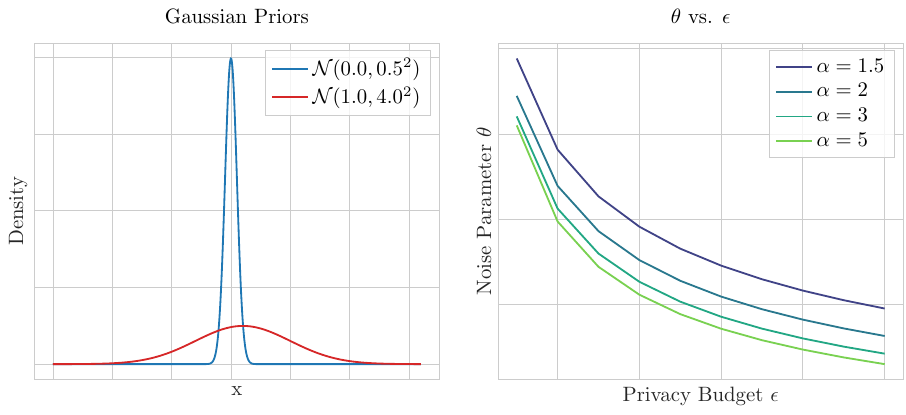}{$\alpha$-Decreasing Case}{fig:theta_eps_gaussian_alpha_dec}

    \vspace{0.6em}
    \gaussianalphafig{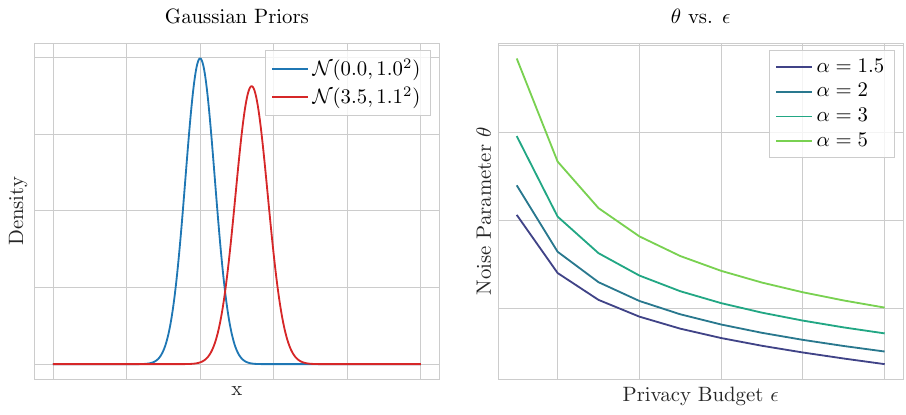}{$\alpha$-Increasing Case}{fig:theta_eps_gaussian_alpha_inc}
    \caption{Two Gaussian-prior examples illustrating how the required noise parameter $\theta$ changes with the R\'{e}nyi order $\alpha$ across the privacy-budget sweep.}
    \label{fig:theta_eps_gaussian_alpha}
\end{figure}

%% file: sec/4_gmm.tex
\section{RPP Gaussian Mixture Models}\label{sec:gmm_rpp}
Our proposed approaches in Section~\ref{sec:gaussian_rpp} established a foundational mechanism for a single Gaussian prior to achieve $(\epsilon, \alpha)$-$\textit{R\'{e}nyi Pufferfish Privacy}$. 
While the universal importance of the Gaussian measures is evident in various fields, there are still some data distributions that cannot be accurately represented by a single Gaussian distribution. 
For these data, our approaches with single Gaussian priors are not suitable, and it's critical to propose a more general mechanism.  
To bridge this gap, we extend our framework to Gaussian Mixture Models (GMMs) in this section.

\subsection{GMM Fitting}
\label{sec:gmm_fitting}

\begin{figure*}[!ht]
    \centering
    \subfloat[Adult\label{fig:gmm_adult}]{%
        \begin{minipage}[t]{0.32\textwidth}
        \centering
        \begin{adjustbox}{width=\linewidth,max totalheight=0.23\textheight,keepaspectratio}
            \includegraphics[width=\linewidth]{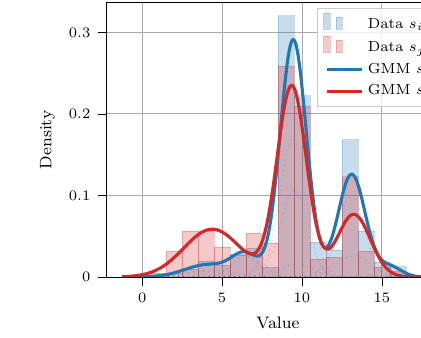}
        \end{adjustbox}
        \end{minipage}%
    }
    \hfill
    \subfloat[Education\label{fig:gmm_education}]{%
        \begin{minipage}[t]{0.32\textwidth}
        \centering
        \begin{adjustbox}{width=\linewidth,max totalheight=0.23\textheight,keepaspectratio}
            \includegraphics[width=\linewidth]{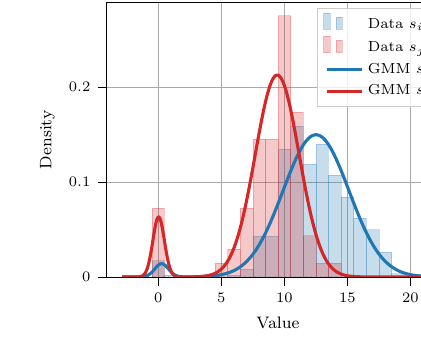}
        \end{adjustbox}
        \end{minipage}%
    }
    \hfill
    \subfloat[Heart Disease\label{fig:gmm_heart}]{%
        \begin{minipage}[t]{0.32\textwidth}
        \centering
        \begin{adjustbox}{width=\linewidth,max totalheight=0.23\textheight,keepaspectratio}
            \includegraphics[width=\linewidth]{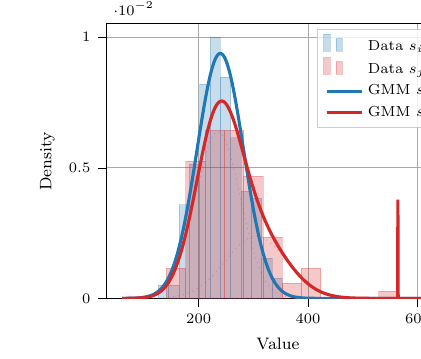}
        \end{adjustbox}
        \end{minipage}%
    }
    \caption{The three real-world datasets are chosen to span demographic, education, and healthcare domains and to illustrate diverse non-Gaussian query distributions. The figure shows the fitted GMM densities on these datasets, where the solid lines represent the estimated mixture densities and the bars represent the empirical distributions.}
    \label{fig:all_gmm_fitting}
\end{figure*}

For query distributions that cannot be accurately represented by a single Gaussian, we approximate each conditional prior with a one-dimensional Gaussian mixture model (GMM). Given empirical query samples
$\mathcal{D}_s=\{x_n\}_{n=1}^{N}$ associated with secret $s$, we fit
\begin{align*}
    P(x|s)
    =
    \sum_{k=1}^{K_s} w_{s,k}
    \mathcal{N}\!\left(x; \mu_{s,k}, \sigma_{s,k}^{2}\right),
\end{align*}
where $\sum_{k=1}^{K_s} w_{s,k}=1$ and $w_{s,k}\ge 0$.

We estimate the mixture parameters by maximum likelihood using the expectation-maximization (EM) algorithm, which is a standard and robust procedure for fitting finite Gaussian mixtures.  
Rather than fixing the number of mixture components, we select the model complexity in a data-dependent manner using the Bayesian information criterion (BIC), a widely adopted choice for GMM-based density estimation~\cite{lmta16}. 
Concretely, for each secret $s$ and a prescribed upper bound $K_{\max}$, we fit candidate models with $K \in \{1,\ldots,K_{\max}\}$ and choose
\begin{align*}
    K_s
    =
    \arg\min_{1\le K\le K_{\max}}
    -2\log p(\mathcal{D}_s;\widehat{\Theta}_{K})
    + q_K \log N,
\end{align*}
where $\widehat{\Theta}_{K}$ denotes the maximum-likelihood parameters for $K$-component GMM and $q_K$ is the number of free parameters. 
In our implementation, $K_{\max}$ is additionally capped by the number of samples and the number of distinct observed values, preventing over-parameterization in small-sample or highly discretized regimes.
We plot the original data and the fitted GMMs in three real-world datasets in Figure~\ref{fig:all_gmm_fitting}.

\subsection{Mechanism Design}
We first introduce the optimal transport plan of GMMs, which is used as the foundational step to design the RPP mechanism for GMMs. 
Consider two conditional query distributions associated with secrets $s_i$ and $s_j$, defined as 
\begin{align*}
    P(x|s_i)
    &= \sum_{k=1}^{K}w_{i,k}\mathcal{N}(x; \mu_{i,k}, \sigma_{i,k}^{2}),\\
    P(x'|s_j)
    &= \sum_{l=1}^{L}w_{j,l}\mathcal{N}(x'; \mu_{j,l}, \sigma_{j,l}^{2}),
\end{align*}
where $\sum_{k=1}^{K} w_{i,k} = 1$, $\sum_{l=1}^{L} w_{j,l} = 1$, and all mixture weights are nonnegative.
Under the additive Gaussian mechanism with noise $\mathcal{N}(0, \theta^{2})$, the privatized distributions remains mixtures of Gaussians, 
\begin{align*}
    P(y|s_i)
    &= \sum_{k=1}^{K}w_{i,k}\mathcal{N}(y; \mu_{i,k}, \sigma_{i,k}^{2}+ \theta^{2}),\\
    P(y|s_j)
    &= \sum_{l=1}^{L}w_{j,l}\mathcal{N}(y; \mu_{j,l}, \sigma_{j,l}^{2}+ \theta^{2}).
\end{align*}

The optimal transport problem for GMMs is to find a global mapping that aligns multi-modal densities. 
We construct a hierarchical transport plan in two stages. 
First, we solve a discrete optimal transport problem over mixture weights to determine an optimal mass transfer between mixture components. 
Let 
\begin{align*}
    \mathcal{U}(w_i,w_j)
    =
    \left\{
    \Pi \in \mathbb{R}_{+}^{K \times L}
    \;\middle|\;
    \sum_{l=1}^{L} \pi_{k,l} = w_{i,k},\;
    \sum_{k=1}^{K} \pi_{k,l} = w_{j,l}
    \right\}    
\end{align*}
denote the transport polytope, we compute the coupling
\begin{align*}
    \Pi^{*}
    =
    \arg\min_{\Pi \in \mathcal{U}(w_i,w_j)}
    \sum_{k=1}^{K}\sum_{l=1}^{L}
    \pi_{k,l}\,\mathbf{C}_{k,l},
\end{align*}
where $\Pi^{*} = (\pi_{k,l}^{*})$ and $\mathbf{C}_{k,l}$ is the component-level matching cost.
In one dimension, we use the squared $2$-Wasserstein distance between Gaussian components,
as $\mathbf{C}_{k,l} = (\mu_{i,k}- \mu_{j,l})^{2}+ (\sigma_{i,k}- \sigma_{j,l})^{2}$, which characterizes the parameter-space distance between the $k$-th source component and the $l$-th target component.
Second, for each source-target pair $(k,l)$ with $\pi_{k,l}^{*} > 0$, we define the Monge map $T_{k \to l}(x) = \mu_{j,l}+ \frac{\sigma_{j,l}}{\sigma_{i,k}}(x - \mu_{i,k})$.
The construction yields a structured transport plan that couples mixture components through $\Pi^{*}$ and transports mass via Gaussian-to-Gaussian maps.

Based on this GMM-OT, we design the RPP mechanism for GMMs as follows.

\begin{theorem}\label{theorem:gmm_rpp}
    Given privacy budget $\epsilon>0$ and R\'{e}nyi divergence order $\alpha>1$, for Gaussian Mixture Model priors
    $X|s_i \sim \sum_{k=1}^K w_{i,k} \mathcal{N}(\mu_{i,k}, \sigma_{i,k}^2)$ and $X|s_j \sim \sum_{l=1}^L w_{j,l} \mathcal{N}(\mu_{j,l}, \sigma_{j,l}^2)$, 
    adding Gaussian noise $\mathcal{N}(0,\theta^2)$ with
    \begin{align}
        \mathbb{E}_{\pi_{k,l}^{*}} \Bigg(
        \exp \Bigg\{
        \frac{\alpha-1}
        {2(\sigma_{i,k}^2 + \theta^2 + \alpha(\sigma_{j,l}^2 - \sigma_{i,k}^2))}
        \notag
        \Big(
        \alpha(\mu_{i,k}-\mu_{j,l})^{2} \notag\\
        + \frac{\alpha^2 (\sigma_{j,l}^2 - \sigma_{i,k}^2)^2}
        {(\alpha-1)(\sigma_{i,k}^2+\theta^2)}
        \Big)
        \Bigg\}
        \Bigg)
        \leq \exp((\alpha-1)\epsilon)
    \end{align}
    attains $(\epsilon,\alpha)$-R\'{e}nyi Pufferfish Privacy for all $(s_i,s_j)\in\mathcal{Q}$.
\end{theorem}
\begin{proof}(Proof of Theorem~\ref{theorem:gmm_rpp})
    We upper bound the R\'{e}nyi divergence between the two privatized GMMs by exploiting the optimal transport plan $\pi^*$. 
    Let $P_i(y) = \sum_{k,l} \pi_{k,l}^* \mathcal{N}(y; \mu_{i,k}, \sigma_{i,k}^2+\theta^2)$ and $P_j(y) = \sum_{k,l} \pi^*_{k,l} \mathcal{N}(y; \mu_{j,l}, \sigma_{j,l}^2+\theta^2)$.
    Using the same coupling weights for the two mixtures allows us to compare the two distributions componentwise. We then have
    \begin{align*}
        D_{\alpha}(P_i || P_j)
        &= \frac{1}{\alpha-1} \ln \int P_i(y)^{\alpha}P_j(y)^{1-\alpha}dy \\
        &\leq \frac{1}{\alpha-1} \ln \int \sum_{k,l}\pi_{k,l}^{*} \mathcal{N}(y; \mu_{i,k}, \sigma_{i,k}^2+\theta^2)^{\alpha} \\
        &\quad \times \mathcal{N}(y; \mu_{j,l}, \sigma_{j,l}^2+\theta^2)^{1-\alpha}dy \\
        &= \frac{1}{\alpha-1} \ln \sum_{k,l}\pi_{k,l}^{*} c_{k,l}(\theta).
    \end{align*}
    Here
    \begin{align*}
        c_{k,l}(\theta)
        &:=
        \sqrt{\frac{(\sigma_{i,k}^{2}+\theta^{2})^{1-\alpha}
        ({\sigma_{j,l}}^{2}+\theta^{2})^{\alpha}}
        {\alpha{\sigma_{j,l}}^{2}+(1-\alpha)\sigma_{i,k}^{2}+\theta^{2}}}
        \notag\\
        &\quad \times
        \exp\left(\frac{\alpha(\alpha-1)(\mu_{i,k}-\mu_{j,l})^{2}}
        {2(\alpha{\sigma_{j,l}}^{2}+(1-\alpha)\sigma_{i,k}^{2}+\theta^{2})}\right).
    \end{align*}
    The second inequality follows from Jensen's Inequality and the positive parameter $\frac{1}{\alpha-1}$ when $\alpha >1$, 
    while the closed-form evaluation of the Gaussian integral follows the derivation of Corollary~\ref{corollary:renyi_gaussian_mechanism}.

    Define
    \begin{align*}
        \psi_{k,l}(\theta)
        &:= \frac{\alpha-1}
        {2(\sigma_{i,k}^2 + \theta^2 + \alpha(\sigma_{j,l}^2 - \sigma_{i,k}^2))}
        \notag\\
        &\quad \times
        \left(
        \alpha(\mu_{i,k}-\mu_{j,l})^{2}
        + \frac{\alpha^2 (\sigma_{j,l}^2 - \sigma_{i,k}^2)^2}
        {(\alpha-1)(\sigma_{i,k}^2+\theta^2)}
        \right).
    \end{align*}
    Then
    \begin{align*}
        c_{k,l}(\theta)
        &\leq \exp\!\big(\psi_{k,l}(\theta)\big),
    \end{align*}
    and hence
    \begin{align*}
        D_{\alpha}(P_i || P_j)
        &\leq \frac{1}{\alpha-1} \ln \sum_{k,l}\pi_{k,l}^{*}\exp\!\big(\psi_{k,l}(\theta)\big).
    \end{align*}

    Therefore, to attain pufferfish privacy it is sufficient to require
    \begin{align*}
        \frac{1}{\alpha-1} \ln \sum_{k,l}\pi_{k,l}^{*}\exp\!\big(\psi_{k,l}(\theta)\big)
        &\leq \epsilon,
    \end{align*}
    which is equivalent to
    \begin{align*}
        \sum_{k,l}\pi_{k,l}^{*}\exp\!\big(\psi_{k,l}(\theta)\big)
        &\leq \exp((\alpha-1)\epsilon).
    \end{align*} 
    This is exactly the condition stated in the theorem.
\end{proof}

Theorem~\ref{theorem:gmm_rpp} gives a sufficient condition to attain $(\epsilon,\alpha)$-\textit{R\'{e}nyi Pufferfish Privacy} for Gaussian mixture priors, which can fit any data distributions.

\newcommand{\gmmthetafig}[3]{
    \subfloat[#2\label{#3}]{%
        \begin{minipage}[t]{0.48\columnwidth}
        \centering
        \begin{adjustbox}{width=\linewidth,max totalheight=0.17\textheight,keepaspectratio}
        \includegraphics[width=\linewidth]{figs/evaluation/#1.pdf}%
        \end{adjustbox}
        \end{minipage}%
    }
}

Figure~\ref{fig:gmm_query_sweeps} shows the corresponding privacy-budget sweep
for GMM priors and compares our mixture-aware calibration with the additive
baseline~\cite{cam24}. The required noise again decreases
monotonically with $\epsilon$ for both \textsc{Raw} and \textsc{Mean} queries,
indicating that the GMM mechanism preserves the same privacy-budget trend as in
the Gaussian case. Our method also consistently attains a smaller $\theta$ than
the baseline throughout the sweep, showing that
Theorem~\ref{theorem:gmm_rpp} provides a tighter calibration when the
secret-conditioned query distributions are better modeled by Gaussian
mixtures.

\begin{figure}[!t]
    \centering
    \gmmthetafig{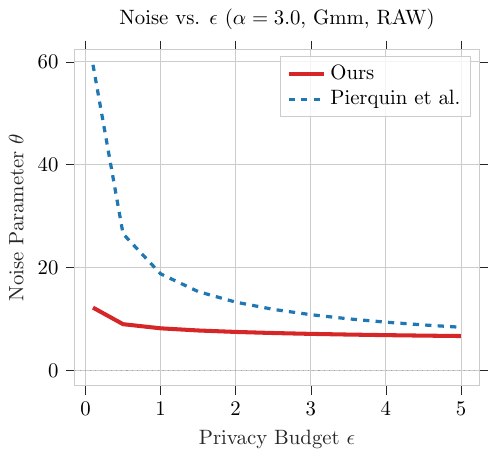}{\textsc{Raw}}{fig:theta_eps_gmm_raw}
    \hfill
    \gmmthetafig{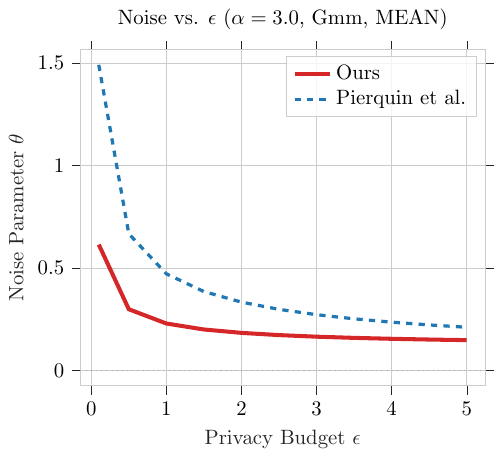}{\textsc{Mean}}{fig:theta_eps_gmm_mean}
    \caption{Noise parameter $\theta$ versus privacy budget $\epsilon$ for GMM-prior queries with $\alpha=3.0$.}
    \label{fig:gmm_query_sweeps}
\end{figure}

%% file: sec/5_experiment.tex
\section{Experiments}\label{sec:experiments}

We empirically evaluate our proposed R\'{e}nyi Pufferfish Privacy mechanisms for Gaussian priors and GMM priors. 
For release $Y$, query output $X$, and zero-mean Gaussian noise $N_\theta$, we have the mean square error (MSE) as $\mathbb{E}[(Y-X)^2] = \mathbb{E}[N_\theta^2] = \theta^2$.
That is, a lower $\theta$ indicates that the mechanism attains the same $(\epsilon, \alpha)$-RPP requirement with less perturbation.
In this work, we use the noise parameter $\theta$ to evaluate the utility of the mechanism.

\newcommand{\evaltheta}[3]{
    \subfloat[#2\label{#3}]{%
        \begin{minipage}[t]{0.238\textwidth}
        \centering
        \begin{adjustbox}{width=\linewidth,max totalheight=0.19\textheight,keepaspectratio}
        \includegraphics[width=\linewidth]{figs/evaluation/#1.pdf}%
        \end{adjustbox}
        \end{minipage}%
    }
}

\newcommand{\alphatheta}[3]{
    \subfloat[#2\label{#3}]{%
        \begin{minipage}[t]{0.32\textwidth}
        \centering
        \begin{adjustbox}{width=\linewidth}
        \includegraphics[width=\linewidth]{figs/alpha/#1.pdf}%
        \end{adjustbox}
        \end{minipage}%
    }
}

\subsection{Experimental Setup}

{\bf Baseline.} We compare our closed-form Gaussian-prior and GMM-prior calibration rules against the RPP additive-noise baseline of Pierquin et al.~\cite{cam24}. 

{\bf Datasets.} The experiments cover three real-world datasets in the UCI machine learning repository, including \textit{Adult}~\cite{adult}, \textit{Heart Disease}~\cite{heart}, and \textit{Student Performance}~\cite{student}. 
Across the three datasets, we consider a privacy-preserving release setting in which one attribute is designated as sensitive, and another is released. 
For the \textit{Adult} dataset, \texttt{race} is treated as the sensitive attribute and \texttt{education-num} as the released attribute, with the secrets defined as ``\texttt{race = White}'' and ``\texttt{race = Other}''. 
For the \textit{Heart Disease} dataset, \texttt{slope} is treated as the sensitive attribute and \texttt{oldpeak} as the released attribute, with the secrets defined as ``\texttt{slope = 1}'' and ``\texttt{slope = 3}''. 
For the \textit{Student Performance} dataset, \texttt{schoolsup} is treated as the sensitive attribute and \texttt{G3} as the released attribute, with the secrets defined as ``\texttt{schoolsup = no}'' and ``\texttt{schoolsup = yes}''. 
In each case, the goal is to guarantee $(\epsilon,\alpha)$-statistical indistinguishability between the corresponding pair of secrets.

{\bf Query.} For each dataset, we evaluate two query families, Statistical Query (\textsc{Raw}, \textsc{Mean}) and ML Query (\textsc{BNN}, \textsc{GP}). 
\textsc{Raw} corresponds to releasing the one-dimensional raw sample itself. 
\textsc{Mean} corresponds to computing the mean on each secret-conditioned subset and approximating its output distribution by a univariate Gaussian using the sample mean and standard error. 
\textsc{BNN} corresponds to training a variational Bayesian neural network separately on each secret-conditioned subset and using the predictive distribution as the query output. 
and \textsc{GP} corresponds to training a Gaussian process regressor on each secret-conditioned subset and using its predictive Gaussian as the query output.

We organize the results along two standard privacy axes. 
First, we compare $\theta$ of our mechanism and baseline under an amount of privacy budget $\epsilon$ while fixing $\alpha=3.0$. 
Second, we plot $\theta$ of our mechanism under different R\'{e}nyi order $\alpha$.

\subsection{Utility}

Figures~\ref{fig:theta_eps_adult_queries},
\ref{fig:theta_eps_heart_queries}, and
\ref{fig:theta_eps_student_queries} compare our prior-aware calibration with
the additive-noise baseline of Pierquin et al.~\cite{cam24} on the three
real-world datasets. Across all datasets and all four query families, our
mechanism consistently requires a smaller $\theta$ than the baseline for the same
$(\epsilon,\alpha)$ target, indicating better utility. This trend is
particularly clear for the \textsc{Raw} and \textsc{Mean} queries, where the
query distributions are often visibly non-unimodal and thus benefit from the
GMM-based calibration in Theorem~\ref{theorem:gmm_rpp}. For the model-output
queries, the same advantage remains under the Gaussian-prior calibration of
Theorem~\ref{theorem:closed_form_mechanism}. Excluding regimes in which our
mechanism requires no additional noise, the proposed method achieves an average
noise reduction of $48.9\%$ relative to the baseline. Overall, the improvement
is not limited to a single dataset or query family; rather, it appears as a
systematic downward shift of our curves relative to the baseline. When our
method reaches zero noise earlier than the baseline, it indicates that
prior-aware calibration can avoid unnecessary perturbation once the two
secret-conditioned distributions already substantially overlap.

\begin{figure*}[!ht]
    \centering
    \evaltheta{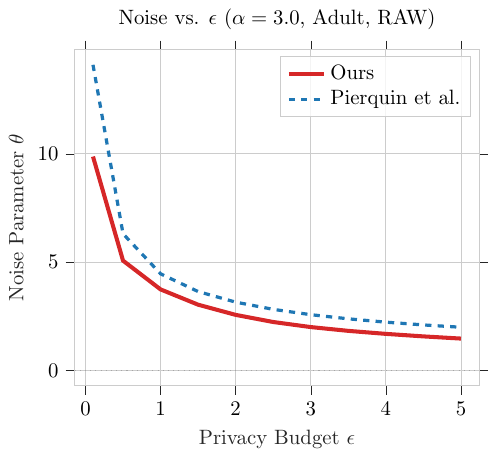}{\textsc{Raw}}{fig:theta_eps_adult_raw}
    \hfill
    \evaltheta{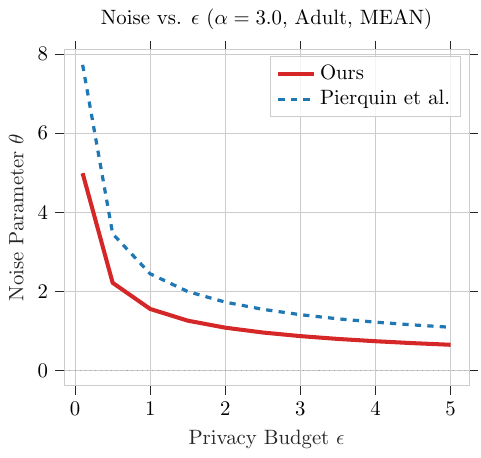}{\textsc{Mean}}{fig:theta_eps_adult_mean}
    \hfill
    \evaltheta{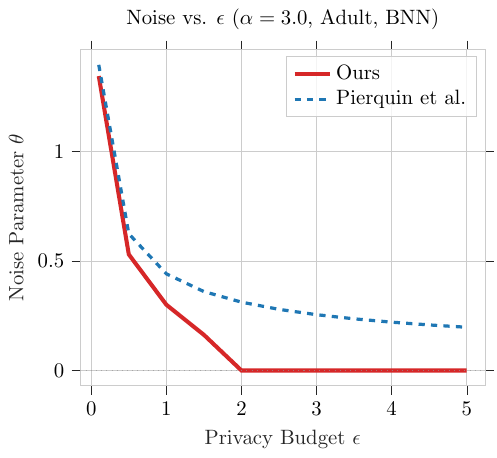}{\textsc{BNN}}{fig:theta_eps_adult_bnn}
    \hfill
    \evaltheta{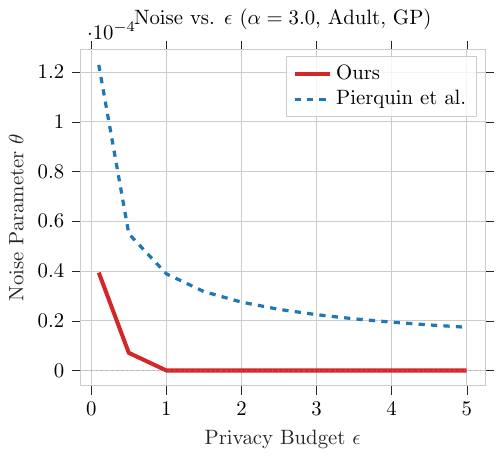}{\textsc{GP}}{fig:theta_eps_adult_gp}
    \caption{\textit{Adult} is chosen as a representative census-style demographic dataset, where the sensitive attribute is \texttt{race}, the released attribute is \texttt{education-num}, and the secrets are $s_i=$``\texttt{race = White}'' and $s_j=$``\texttt{race = Other}''. The figure shows the required noise parameter $\theta$ versus privacy budget $\epsilon$ with $\alpha=3.0$, where the four panels correspond to \textsc{Raw}, \textsc{Mean}, \textsc{BNN}, and \textsc{GP} queries.}
    \label{fig:theta_eps_adult_queries}
\end{figure*}

\begin{figure*}[!ht]
    \centering
    \evaltheta{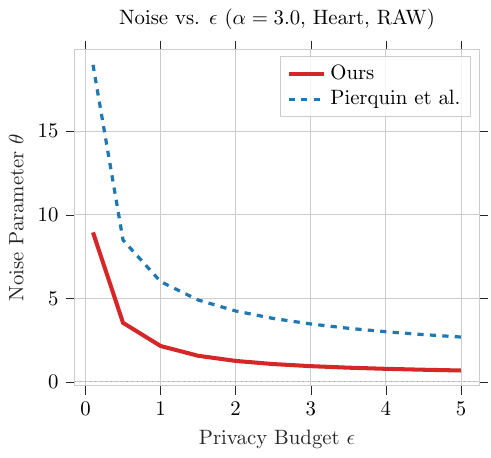}{\textsc{Raw}}{fig:theta_eps_heart_raw}
    \hfill
    \evaltheta{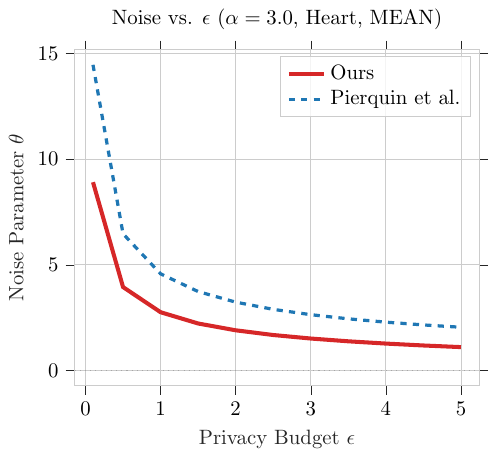}{\textsc{Mean}}{fig:theta_eps_heart_mean}
    \hfill
    \evaltheta{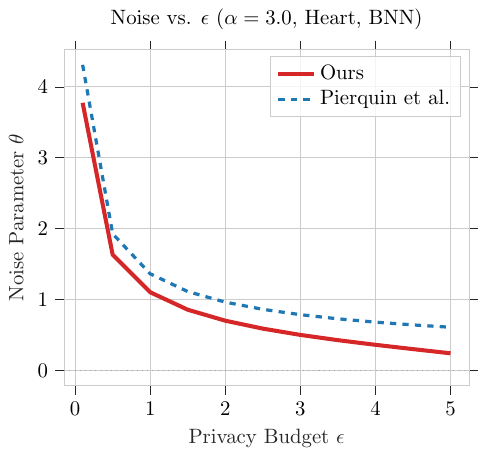}{\textsc{BNN}}{fig:theta_eps_heart_bnn}
    \hfill
    \evaltheta{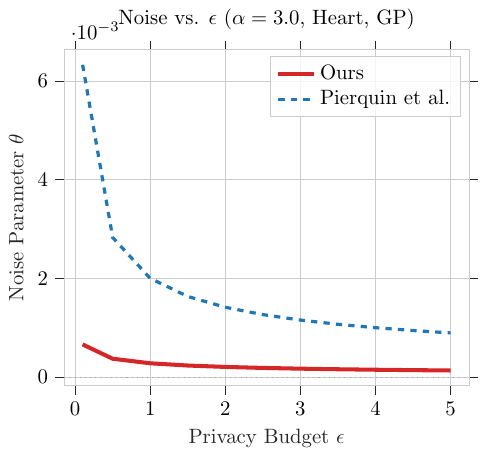}{\textsc{GP}}{fig:theta_eps_heart_gp}
    \caption{\textit{Heart Disease} is chosen as a representative healthcare dataset, where the sensitive attribute is \texttt{slope}, the released attribute is \texttt{oldpeak}, and the secrets are $s_i=$``\texttt{slope = 1}'' and $s_j=$``\texttt{slope = 3}''. The figure shows the required noise parameter $\theta$ versus privacy budget $\epsilon$ with $\alpha=3.0$, where the four panels correspond to \textsc{Raw}, \textsc{Mean}, \textsc{BNN}, and \textsc{GP} queries.}
    \label{fig:theta_eps_heart_queries}
\end{figure*}

\begin{figure*}[!ht]
    \centering
    \evaltheta{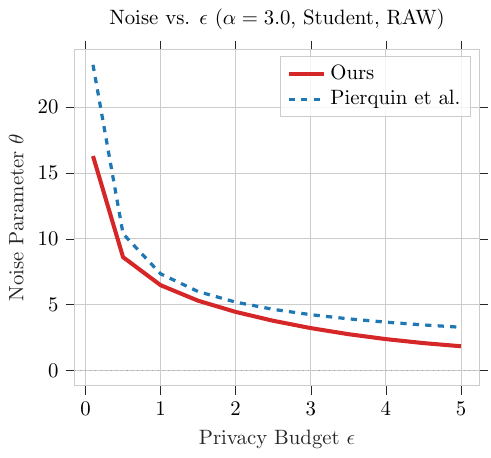}{\textsc{Raw}}{fig:theta_eps_student_raw}
    \hfill
    \evaltheta{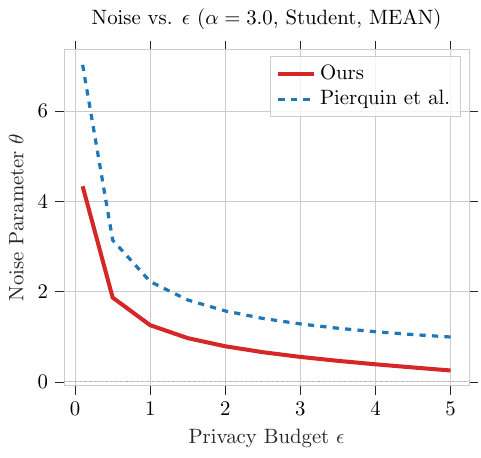}{\textsc{Mean}}{fig:theta_eps_student_mean}
    \hfill
    \evaltheta{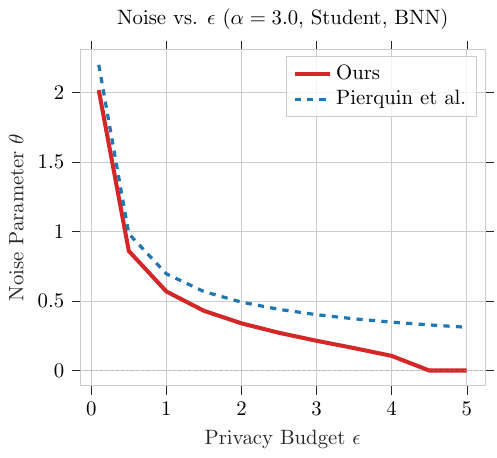}{\textsc{BNN}}{fig:theta_eps_student_bnn}
    \hfill
    \evaltheta{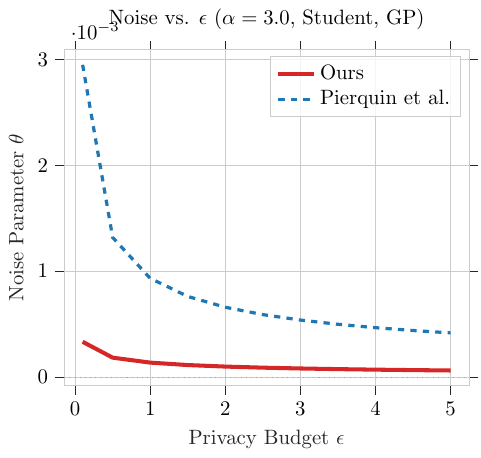}{\textsc{GP}}{fig:theta_eps_student_gp}
    \caption{\textit{Student Performance} is chosen as a representative education-performance dataset, where the sensitive attribute is \texttt{schoolsup}, the released attribute is \texttt{G3}, and the secrets are $s_i=$``\texttt{schoolsup = no}'' and $s_j=$``\texttt{schoolsup = yes}''. The figure shows the required noise parameter $\theta$ versus privacy budget $\epsilon$ with $\alpha=3.0$, where the four panels correspond to \textsc{Raw}, \textsc{Mean}, \textsc{BNN}, and \textsc{GP} queries.}
    \label{fig:theta_eps_student_queries}
\end{figure*}

\subsection{Monotonicity}

We next examine how the calibrated noise changes with the privacy parameters.
For the budget sweep, the \textsc{Ours} curves in
Figures~\ref{fig:theta_eps_adult_queries},
\ref{fig:theta_eps_heart_queries}, and
\ref{fig:theta_eps_student_queries} are all monotonically decreasing in
$\epsilon$, which is consistent with
Proposition~\ref{prop:theta_decrease_epsilon}. This monotonicity holds for
scalar queries as well as model-output queries. In some easier regimes, such as
Adult \textsc{GP} and Adult \textsc{BNN}, the curves flatten at zero after a
moderate privacy budget, which still agrees with the predicted non-increasing
behavior.

Figures~\ref{fig:theta_eps_realworld_alpha} and~\ref{fig:theta_eps_gaussian_alpha}
further show the monotonicity with respect to the R\'{e}nyi order $\alpha$. On
all three real-world BNN queries, fixing $\epsilon$ and increasing $\alpha$
consistently shifts the curves upward, matching the increasing regime described
in Proposition~\ref{prop:theta_monotonicity_alpha}. The ordering of the curves is
stable across datasets: smaller $\alpha$ leads to uniformly lower noise, while
larger $\alpha$ keeps the curves elevated over a wider range of $\epsilon$.
This pattern confirms that stronger tail-sensitive privacy requirements demand
more perturbation, exactly as predicted by the theory.

\begin{figure*}[!t]
    \centering
    \alphatheta{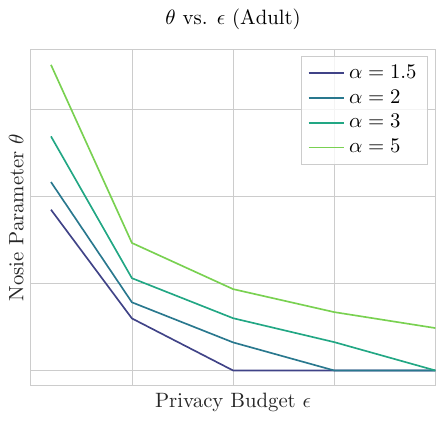}{Adult}{fig:theta_eps_alpha_adult}
    \hfill
    \alphatheta{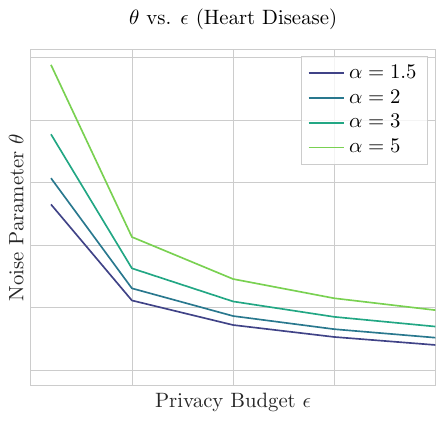}{Heart Disease}{fig:theta_eps_alpha_heart}
    \hfill
    \alphatheta{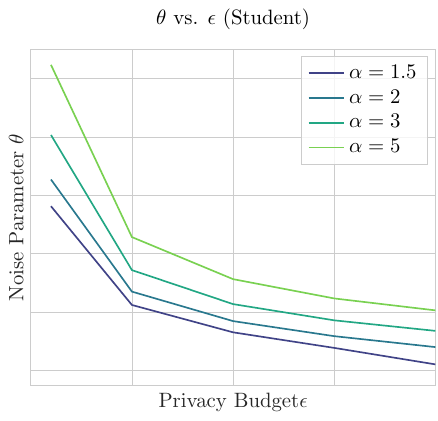}{Student}{fig:theta_eps_alpha_student}
    \caption{The three real-world datasets including \textit{Adult}, \textit{Heart Disease}, and \textit{Student Performance}. \textit{Adult} releases attribute \texttt{education-num} while preserving sensitive attribute \texttt{race}, where secrets $s_i=$``\texttt{race = White}'' and $s_j=$``\texttt{race = Other}''; \textit{Heart Disease} releases attribute \texttt{oldpeak} while preserving sensitive attribute \texttt{slope}, and secrets $s_i=$``\texttt{slope = 1}'' and $s_j=$``\texttt{slope = 3}''; and \textit{Student Performance} releases attribute \texttt{G3} while preserving sensitive attribute \texttt{schoolsup}, where secrets $s_i=$``\texttt{schoolsup = no}'' and $s_j=$``\texttt{schoolsup = yes}''. The figure shows how the $\theta$--$\epsilon$ curves of our \textsc{BNN} query calibration change under different fixed R\'{e}nyi orders $\alpha \in \{1.5,2.0,3.0,5.0\}$.}
    \label{fig:theta_eps_realworld_alpha}
\end{figure*}



The model-output experiments further evaluate $\alpha$-sensitivity for BNN and GP outputs. 
They follow the same qualitative pattern as the scalar query experiments: stronger R\'{e}nyi orders require larger perturbations, while the
proposed calibration remains below the baseline across datasets
and model families.



\subsection{Summary}

The experiments support three main conclusions. 
First, compared with the additive-noise baseline, our mechanisms consistently require less noise, which directly improves data utility under the same $(\epsilon, \alpha)$-\textit{R\'{e}nyi Pufferfish Privacy} guarantee. 
Second, the proposed framework is not limited to Gaussian priors, through the GMM-based construction, it also applies to general non-Gaussian and multimodal query distributions, extending the method beyond the unimodal setting. 
Third, the key properties established in the theoretical analysis are reflected by the empirical results, as the utility advantage and the monotonic trends with respect to $\epsilon$ and $\alpha$ consistently match the predictions of our theorems and propositions.

%% file: sec/6_conclusion.tex
\section{Conclusion}\label{sec:conclusion}
This paper studies Gaussian mechanisms for R\'enyi Pufferfish Privacy (RPP) under single Gaussian and Gaussian-mixture priors. 
For single Gaussian priors, we derived the exact R\'enyi divergence after Gaussian perturbation, established the monotonicity of the privacy loss with respect to the noise variance, and presented both a numerical calibration algorithm to determine the optimal noise variance.
We also provide a closed-form sufficient condition along with its monotonicity with respect to privacy budget $\epsilon$ and R\'{e}nyi order $\alpha$. 
To move beyond the general setting, we further modeled secret-conditioned outputs with Gaussian mixture models and developed an optimal-transport-based sufficient condition for RPP under mixture priors. 
Experimental results on three UCI datasets with statistical and model-output queries showed that the proposed prior-aware mechanisms consistently require less noise than a recent additive-noise baseline, leading to improved utility under the same $(\epsilon,\alpha)$-RPP guarantee.